\documentclass[12pt,a4paper]{article}
\usepackage[top=1 in, bottom=1 in, left=1 in, right=1 in, letterpaper]{geometry}
\usepackage{amssymb,amsmath}
\usepackage{xspace}
\usepackage{algorithm}
\usepackage[noend]{algorithmic}
\usepackage[english]{babel}
\usepackage{fancyhdr}

\newtheorem{theorem}{Theorem}
\newtheorem{lemma}{Lemma}

\newtheorem{claim}{Claim}
\newenvironment{proof}[1][Proof]{\textbf{#1:}\\}{\hfill\rule{1.5mm}{1.5mm}\\}

\title{Speed Scaling on Parallel Processors with Migration}

\author{
Eric Angel$^1$, Evripidis Bampis$^2$, Fadi Kacem$^1$, Dimitrios Letsios$^1$ \\
\\
$^1$IBISC, Universit\'{e} d'\'{E}vry, France\\
\{eric.angel, fadi.kacem, dimitris.letsios\}@ibisc.univ-evry.fr\\
\\
$^2$LIP6, Universit\'{e} Pierre et Marie Curie, France\\
Evripidis.Bampis@lip6.fr\\
}

\begin{document}

\maketitle

\begin{abstract}
We study the problem of scheduling a set of jobs with release dates, deadlines and processing requirements (or works), on parallel
speed-scaled processors so as to minimize the total energy consumption.
 We consider that both preemption and migration of jobs are allowed.
An exact polynomial-time algorithm has been proposed for this
problem, which is based on the Ellipsoid algorithm.
Here, we formulate the problem as a convex program
and we propose a simpler polynomial-time combinatorial algorithm which is based on a reduction
to the maximum flow problem. Our algorithm runs in $O(nf(n)logP)$ time,
 where $n$ is the number of jobs, $P$ is the range of all possible
values of processors' speeds divided by the desired accuracy and $f(n)$ is the complexity of computing a maximum flow in a layered graph
with $O(n)$ vertices.
Independently, Albers et al. \cite{AAG11}
proposed an $O(n^2f(n))$-time algorithm  exploiting the
same relation with the maximum flow problem.
We  extend our algorithm to the  multiprocessor
speed scaling problem with migration
where the  objective is the minimization
of the  makespan under a budget of energy.
\end{abstract}

\section{Introduction}

Energy consumption is a major issue in our days. Great efforts  are devoted to the reduction of energy dissipation in computing environments ranging from small portable devices to large data centers. From an algorithmic point of view, new challenging optimization problems are studied, in which the energy consumption is taken into account  as a constraint or as the optimization goal itself
(for recent reviews see  \cite{A10, A11}).
This later approach has been adopted in the seminal paper of Yao et al. \cite{YDS95}, where a set of independent jobs with release dates and deadlines have to be scheduled on a single processor so that the total energy is minimized, under the so-called {\em speed-scaling} model where the processor may run at variable speeds. Under this model, if the speed of a processor is $s$ then the power consumption is  $s^{\alpha}$, where $\alpha >1$ is a constant, and the energy consumption is the power integrated over time.
\newline
\newline
\noindent
{\bf Single processor case.}
Yao et al. proposed in \cite{YDS95}, an optimal off-line algorithm, known as the YDS algorithm according to the initials of the authors, for the problem with  preemption, i.e. where the execution of a job may be interrupted and resumed later on. In the same work, they initiated the study of online algorithms for the problem, introducing the Average Rate (AVR) and the Optimal Available (OA) algorithms. Bansal et al. \cite{BKP07} proposed a new online algorithm, the BKP algorithm according to the authors' initials, which improves the competitive ratio of OA for large values of $\alpha$.
\newline
\newline
\noindent
{\bf Multiprocessor case.} There are two variants of the model: the first variant allows the preemption of the jobs but not their migration. We call this variant, the {\em non-migratory} variant. This means that a job may be interrupted and resumed later on, on the same processor, but it is not allowed to continue its execution on a different processor. In the second variant, the {\em migratory} variant, both the preemption and the migration of the jobs are allowed. In \cite{AMSa07}, Albers et al. considered the non-migratory problem of minimizing the total energy consumption given that the jobs have release dates and deadlines. For unit-work jobs, they proposed a polynomial time algorithm when the deadlines of jobs are {\em agreeable}. When the release dates and deadlines of jobs are arbitrary, they proved that the problem becomes NP-hard even for unit-size jobs and proposed approximation algorithms with constant approximation ratios for the off-line version of the problem. A generic reduction is given by Greiner et al. (see \cite{GNS09}) transforming a $\beta$-approximation algorithm for the single-processor problem to a  $\beta B_{\alpha}$-approximation algorithm for the multi-processor non-migratory problem, where $B_{\alpha}$ is the $a$-th Bell number. Also, they showed that a $\beta$-approximation for multiple processors with migration yields a deterministic $\beta B_{\alpha}$-approximation algorithm for multiple processors without migration.

For the migratory variant, Chen et al., in \cite{CHCYPK04}, were the first to study the speed scaling problem of minimizing the energy consumption on $m$ processors with migration. In fact, they proposed a simple algorithm for the case where jobs have common release dates and deadlines. In \cite{BG08}, Bingham and Greenstreet proposed a polynomial-time algorithm for the general problem where each job has an arbitrary work, a release date and a
deadline, and the power function is any convex function. Their algorithm is based on the use of the
Ellipsoid method (see \cite{NNN94}).
Since the Ellipsoid algorithm is not used in practice,
it was an open problem to define a faster combinatorial
algorithm.
When preparing the current version of this paper, it came to our
knowledge that Albers et al. \cite{AAG11} considered the same
problem and presented an optimal $O(n^2f(n))$-time
combinatorial algorithm, where $n$ is the number of jobs
and $f(n)$ the complexity of finding a maximum
flow in a layered graph with $O(n)$ vertices. Notice that in \cite{AAG11},
nothing is mentioned about the exact complexity
of the algorithm, except of course of
its clear polynomiality. They also  extended  the analysis of the single
processor OA and AVR online algorithms
 to the multiprocessor case
with migration.
\newline
\newline
\noindent
{\bf Multicriteria minimization.}
In general, minimizing the energy consumption is in conflict with
the increase of the performance of many computing devices. Hence,
a series of papers adresses this problem in a multicriteria
context.
In \cite{PUW08}, Pruhs et al. were the first to study the problem of optimizing a time-related objective function with a budget of energy. Their objective was to minimize the sum of flow times and they presented a polynomial time algorithm for the case of unit-work jobs. To prove that their algorithm is optimal, they formulated the problem as a convex program and they applied the well-known
Karush-Kuhn-Tucker (KKT) conditions to get necessary conditions for optimality. In \cite{AF07},
Albers and  Fujiwara studied the problem of minimizing the sum of flow times plus energy instead of having an energy budget, which gives rise
to an alternative way of combining  the optimization of two conflicting
criteria. For unit-work jobs, they proposed online algorithms and an exact polynomial-time algorithm. In \cite{CCLLMW07},  Chan et al. proposed an online algorithm to minimize the energy consumption and among the schedules with the minimum energy they tried to find the one with the maximum throughput. Assuming that there is an upper bound on the processor's speed, they established constant-factor competitiveness both in terms of energy and throughput.
\newline
\newline
\noindent
{\bf Our contribution and organization of the paper.}
We consider the multiprocessor migratory scheduling problem
with the objective of minimizing the energy consumption.
In Section 3, we give the first convex programming
formulation of the problem and in Section 4, we apply,
for the first time, the well known KKT conditions.
In this way, we obtain a set of properties that
need to be satisfied by any optimal schedule.
Then in Section 5, we propose an optimal algorithm in the case where the jobs have release dates, deadlines and the power function is of the form $s^{\alpha}$. The time
complexity of our algorithm, which we call BAL, is in $O(nf(n)\log P)$, where $n$ is the number of jobs, $P$ is the range of all possible
values of processors' speed divided by the desired accuracy and $f(n)$ is the complexity of computing a maximum flow in a layered graph with $O(n)$ vertices.
We also give a brief description of the relation of our algorithm
and the one of Albers et al. \cite{AAG11}, as well as the analysis of their
algorithm's complexity.
Finally in Section 6, we extend BAL to obtain an optimal algorithm for the problem of makespan minimization
with a budget of energy.

\section{Preliminaries}

Let $\mathcal{J}=\{j_1,...,j_n\}$ be a set of jobs. Each job $j_i$ is specified by a work $w_i$, a release date $r_i$ and a deadline $d_i$. We define $span_i=[r_i,d_i]$ and we say that $j_i$ is \textbf{alive} at time $t$ if $t\in span_i$. We also define the density of job $j_i$ as $den_i=w_i/(d_i-r_i)$. We assume a set of $m$ variable-speed homogeneous processors in the sense that they can all, dynamically, change their speeds and have a common speed-to-power function $P(t)=s(t)^\alpha$ where $P(t)$ is the power consumption at time $t$, $s(t)$ is the speed (or frequency) at time $t$ and $\alpha>1$ is a constant. Consider any interval of time $[a,b]$ and a given processor. The amount of work processed by this processor and its energy consumption during $[a,b]$ are $\int_a^{b}s(t)dt$ and $\int_a^{b}s(t)^\alpha dt$, respectively. Hence, if a job is continuously run at a constant speed $s$ during an interval of length $\ell$, then $w=s\cdot\ell$ units of work are completed and an amount of $E=s^\alpha\cdot\ell$ units of energy are consumed. In our setting, preemption and migration of jobs are allowed. That is, the processing of a job may be suspended and resumed later on the same processor or on a different one. Nevertheless, we do not allow parallel execution of a job which means that a job cannot be run simultaneously on two or more processors. We also assume that a continuous spectrum of speeds is available and that there is no upper bound on the speed of any processor. Our objective is to find a feasible schedule that minimizes the total energy consumed by all processors.

We define $\mathcal T = \{t_0 ,\cdots\, t_L\}$ to be the set of release dates and deadlines taken in a non-decreasing order and without duplication. It is clear that $t_0=\min_{j_i\in\mathcal{J}}\{r_i\}$ and $t_L=\max_{j_i\in\mathcal{J}}\{d_i\}$. Let $I_j=[t_{j-1},t_j]$, for $1\leq j\leq L$, and $\mathcal{I}=\{I_1,\cdots, I_L\}$. We denote $|I_j|$ the length of the interval $I_j$. Also, let $A(j)$ be the set of jobs that are alive during $I_j$, i.e. all the jobs $j_i$ with $I_j\subseteq span_i$, and $a_j=|A(j)|$ be the number of jobs in $A(j)$. Given any schedule $\mathcal{S}$, we denote $t_{i,j}$ the total units of time that job $j_i$ is processed during the interval $I_j$ by $\mathcal{S}$. As already mentioned in many other works (see \cite{YDS95} for example), one can show, through a simple exchange argument, that there always exists an optimal schedule in which every job $j_i$ is run at a constant speed $s_i$ and this comes from the convexity of the power function.

Next, we state a problem which is a variation of our problem that we will need throughout our analysis, we call it the {\em Work Assignment Problem} (or WAP) and can be described as follows: Consider a set of $n$ jobs $\mathcal{J}=\{j_1,j_2,\cdots,j_n\}$ and a set of intervals $\mathcal{I}=\{I_1,I_2,\cdots,I_L\}$. Each job can be alive in one or more intervals in $\mathcal{I}$. During each interval $I_j$ there are $m_j$ available processors. Moreover, we are given a value $v$. Our objective is to find whether or not there is a feasible schedule that executes all jobs in $\mathcal{J}$ with constant speed $v$. Recall that a schedule is feasible if and only if each job is executed during its alive intervals and is executed by at most one processor at each time $t$. Preemption and migration of jobs are allowed.
Note that the WAP is almost the $P|r_i,d_i,pmtn|-$ (see \cite{BNS04}) with the difference that, in WAP, not all intervals have the same number of available processors. Therefore, WAP is polynomially solvable by applying a variant of an
algorithm for $P|r_i,d_i,pmtn|-$.

\section{Convex Programming Formulation}
Our problem can be formulated as the following convex program:
\begin{eqnarray}
\min \sum_{j_i\in \mathcal{J}} w_is_i^{\alpha-1}\\ \nonumber \\
\frac{w_i}{s_i} - \sum_{I_j:\;j_i\in A(j)}t_{i,j} \leq 0  & \hspace{2cm} j_i\in \mathcal {J}\\
\sum_{j_i\in A(j)} t_{i,j} - m\cdot|I_j| \leq 0 & \hspace{2cm} 1\leq j\leq L\\
\sum_{j_i\in A(j)} t_{i,j} - a_j\cdot|I_j| \leq 0 & \hspace{2cm} 1\leq j\leq L\\
t_{i,j} - |I_j|\leq 0 & \hspace{2cm} 1\leq j\leq L, \; j_i\in A(j) \\
- t_{i,j}\leq 0 & \hspace{2cm} 1\leq j\leq L, \; j_i\in A(j)\\	
- s_i \leq 0 & \hspace{2cm} j_i\in \mathcal{J}
\end{eqnarray}
Note that the total running time and the total energy consumption of each job $j_i$ is $\frac{w_i}{s_i}$ and $w_is_i^{a-1}$, respectively. Then, the term (1) is the total energy consumed by all jobs which is our objective function and the constraints (2) enforce that $w_i$ amount of work must be executed for each job $j_i$. The constraints (3) enforce that we can use at most $m$ processors for $|I_j|$ units of time during any interval $I_j$. Also, we can use at most $a_j$ processors operating for $|I_j|$ units of time during any interval $I_j$, otherwise we would have parallel execution of a job and this is expressed by (4). The constraints (5) prevent any job $j_i$ from being executed for more than $|I_j|$ units of time during any interval $I_j\subseteq span_i$. Note that constraints (4) and (5) are both needed and none is covered by the other. The constraints (6) and (7) insure the positivity of the variables $t_{i,j}$ and $s_i$, respectively.

The above mathematical program is indeed convex because, as mentioned by other works (e.g. \cite{PUW08}), the objective function and the first constraint are convex while all the other constraints are linear. Since our problem can be written as a convex program, it can be solved in polynomial time by applying the Ellipsoid Algorithm \cite{NNN94}. Nevertheless, the Ellipsoid Algorithm is not used in practice and we would like to construct a faster and less complicated combinatorial algorithm.

At this point, notice that once the speeds of the jobs are computed, by solving the convex program, a further step is needed in order to construct a feasible schedule. This is exactly the feasibility problem $P|r_i,d_i,pmtn|-$.

\section{KKT Conditions}
We apply  the KKT conditions to the above convex program
to obtain necessary conditions for optimality of a feasible
schedule. We also show that these conditions are sufficient for optimality.

Assume that we are given the following convex program:

\begin{eqnarray*}
\min f(x) \\
g_i(x)\leq0 & \hspace{2cm} 1\leq i\leq m \\
x\in \mathbf{R}^n
\end{eqnarray*}

Suppose that the program is strictly feasible, i.e. there is a point $x$ such that $g_i(x)<0$ for all $1\leq i\leq m$, and all functions $g_i$ are differentiable. Let $\lambda_i$ be the dual variable associated with the constraint $g_i(x)\leq0$. The Karush-Kuhn-Tucker (KKT) conditions are:

\begin{eqnarray*}
g_i(x)\leq0 & \hspace{2cm} 1\leq i\leq m \\
\lambda_i\geq0 & \hspace{2cm} 1\leq i\leq m \\
\lambda_ig_i(x)=0 & \hspace{2cm} 1\leq i\leq m \\
\nabla f(x)+\sum_{i=1}^m\lambda_i\nabla g_i(x) =0
\end{eqnarray*}

KKT conditions are necessary and sufficient for solutions $x\in \mathbf{R}^n$ and $\lambda\in \mathbf{R}^m$ to be primal and dual optimal. We refer to the above conditions as primal feasible, dual feasible, complementary slackness and stationarity conditions, respectively.

The following lemma is a direct consequence of the KKT conditions for the convex program of our problem.

\begin{lemma}\label{KKT-lem}
A feasible schedule for our problem is optimal if and only if it satisfies the following properties:
\begin{enumerate}
  \item Each job $j_i$ is executed at a constant speed $s_i$.
  \item If a job $j_i$ is not executed during an interval $I_j\subset span_i$, i.e. $t_{i,j}=0$, then $s_i\leq s_k$ for every job $j_k$ with $I_j\subseteq span_k$ and $t_{k,j}>0$.
  \item If a job $j_i$ has $t_{i,j}=|I_j|$ for an interval $I_j$, then $s_i\geq s_k$ for any job $j_k$ alive during $I_j$ with $t_{k,j}<|I_j|$.
  \item All jobs $j_i$ that are alive during $I_j$ with $0<t_{i,j}<|I_j|$ have equal speeds.
  \item If $a_j \leq m$ during an interval $I_j$, then $t_{i,j}=|I_j|$, for every $j_i$ with $I_j \subseteq span_i$.
\end{enumerate}
\end{lemma}

\begin{proof}
In order to apply the KKT conditions, we need to associate with each constraint a dual variable. Therefore, to each set of constraints from (2) up to (7) we associate the dual variables $\beta_i$, $\gamma_j$, $\delta_j$, $\epsilon_{i,j}$, $\zeta_{i,j}$ and $\eta_{i}$, respectively.

By stationarity conditions, we have that
\begin{eqnarray*}
\nabla\sum_{j_i\in\mathcal{J}} w_is_i^{\alpha-1} + \sum_{j_i\in\mathcal{J}}\beta_i\cdot\nabla\bigg(\frac{w_i}{s_i} - \sum_{I_j:\;j_i\in A(j)}t_{i,j}\bigg)\\
+ \sum_{j=1}^L \gamma_j \nabla\bigg(\sum_{j_i\in A(j)} t_{i,j}-m\cdot|I_j|\bigg) + \sum_{j=1}^L \delta_j \nabla\bigg(\sum_{j_i\in A(j)} t_{i,j} - a_j\cdot|I_j|\bigg)  \\
+ \sum_{j=1}^L \sum_{j_i\in A(j)} \epsilon_{ij} \nabla(t_{i,j} - |I_j|) + \sum_{j=1}^L \sum_{j_i\in A(j)} \zeta_{ij} \nabla(-t_{i,j}) + \sum_{j_i\in \mathcal{J}} \eta_{i} \nabla(-s_i) &=& 0
\end{eqnarray*}
The above equation can be rewritten equivalently as
\begin{eqnarray}
\sum_{j=1}^L\sum_{j_i\in A(j)}\bigg(-\beta_i+\gamma_j+\delta_j+\epsilon_{i,j}-\zeta_{i,j}\bigg)\nabla t_{i,j} \nonumber\\ +\sum_{j_i\in\mathcal{J}}\bigg(\bigg((\alpha-1)w_is_i^{\alpha-2}-\frac{\beta_iw_i}{s_i^2}\bigg)\sum_{I_j:\;j_i\in A(j)} t_{i,j}-\eta_i\bigg)\nabla s_i &=& 0
\end{eqnarray}
Furthermore, complementary slackness conditions imply that
\begin{eqnarray}
     \beta_i\cdot\bigg(\frac{w_i}{s_i} - \sum_{I_j:\;j_i\in A(j)}t_{i,j}\bigg)=0  &\hspace{2cm} j_i\in \mathcal{J}\\
     \gamma_j\cdot\bigg(\sum_{j_i\in A(j)} t_{i,j}-m\cdot|I_j|\bigg)= 0 &\hspace{2cm} 1\leq j\leq L\\
     \delta_j\cdot\bigg(\sum_{j_i\in A(j)} t_{i,j} - a_j\cdot|I_j|\bigg)= 0 &\hspace{2cm} 1\leq j\leq L\\
     \epsilon_{ij}\cdot(t_{i,j} - |I_j|)= 0 &\hspace{2cm} 1\leq j\leq L, \; j_i\in A(j) \\
     \zeta_{ij}\cdot(-t_{i,j})= 0 &\hspace{2cm} 1\leq j\leq L, \; j_i\in A(j) \\	
     \eta_i\cdot(-s_i) = 0 &\hspace{2cm} j_i\in\mathcal{J}
\end{eqnarray}
We can safely assume that there are no jobs with zero work because we may
treat such jobs as if they did not exist. So, for any job $j_i$ it holds that $s_i>0$ and $\sum_{I_j\subseteq span_i}t_{i,j}>0$. Then, (14) implies that $\eta_i=0$. We set the coefficients of the partial derivatives $\nabla s_i$ and $\nabla t_{i,j}$ equal to zero so as to satisfy the stationarity conditions.
Thus, (8) gives that $\beta_i=(\alpha-1)s_i^\alpha$ for each job $j_i\in\mathcal{J}$ and
\begin{equation} (\alpha-1)s_i^\alpha=\gamma_j+\delta_j+\epsilon_{i,j}-\zeta_{i,j} \end{equation}
for each $j_i\in\mathcal{J}$ and $I_j\subseteq span_i$.
Now, for each interval $I_j$ we have the following cases:
\newline

\textbf{Case 1:} $a_j>m$ \newline
In this case, it is obvious that all processors operate during the whole interval in any optimal schedule. Because of (11), $\delta_j=0$.
We consider the following subcases on the execution time of any job $j_i\in A(j)$:
\begin{enumerate}
  \item Subcase A: $0<t_{i,j}<|I_j|$ \\
    Stationarity conditions (12), (13) imply that $\epsilon_{i,j}=\zeta_{i,j}=0$. As a result, (15) can be written as
    \begin{equation} (\alpha-1)s_i^\alpha=\gamma_j\nonumber\end{equation}
    The variable $\gamma_j$ is specific for each interval and as a result, all jobs of this subcase  have the same speed throughout the whole schedule. We denote this speed $v_j$ for each interval $I_j$.
  \item Subcase B: $t_{i,j}=|I_j|$ \\
    In this case, by (13) and (15), we get that
    \begin{equation} (\alpha-1)s_i^\alpha=\gamma_j+\epsilon_{i,j}\end{equation}
    Hence, all jobs of this kind have $s_i\geq v_j$.
  \item Subcase C: $t_{i,j}=0$ \\
    Which means, by (12) and (15) that
    \begin{equation} (\alpha-1)s_i^\alpha=\gamma_j-\zeta_{i,j}\end{equation}
    and thus, $s_i\leq v_j$.
\end{enumerate}

\textbf{Case 2:} $a_j<m$
\newline
In this case, each job in $A(j)$ is executed throughout the whole interval $I_j$, in every optimal schedule. This argument comes from the convexity of speed to power function. Therefore, each job $j_i\in A(j)$ has $\zeta_{i,j}=0$. Moreover since fewer than $m$ processors are used we have that $\gamma_j=0$. That is, for each $j_i\in I_j$ we have $(\alpha-1)s_i^\alpha=\delta_j+\epsilon_{i,j}$. By this set of equations, we cannot establish any strong relation between the speed of the  jobs that are alive during an interval $I_j$.
\newline

\textbf{Case 3:} $a_j=m$
\newline
This case can be handled exactly as the previous one with the difference that $\gamma_j\geq0$ and thus, we get that $(\alpha-1)s_i^\alpha=\gamma_j+\delta_j+\epsilon_{i,j}$.
\newline

Given a solution of the convex program that satisfies the KKT conditions, we derived some relations between the primal variables. Based on them, we defined some structural properties of any optimal schedule. These properties are necessary for optimality and we show that they are also sufficient because any schedule that satisfies these properties is optimal.

Assume for the sake of contradiction that there is a schedule $A$, that satisfies the properties of lemma 1, which is not optimal and let $B$ be an optimal schedule. We denote $E^X$, $s_i^X$ and $t_{i,j}^X$ the energy consumption, the speed of job $j_i$ and the total execution time of job $j_i$ during the interval $I_j$ in schedule $X$, respectively. Then, $E^A>E^B$. Let $S$ be the set of jobs $j_i$ with $s_i^A>s_i^B$. Clearly, there is at least one job $j_k$ such that $s_k^A>s_k^B$, otherwise $A$ would not consume more energy than $B$. So, $S\neq\emptyset$. By definition of $S$,
\begin{displaymath}\sum_{j_i\in S}\sum_{I_j:j_i\in A(j)}t_{i,j}^A<\sum_{j_i\in S}\sum_{I_j:j_i\in A(j)}t_{i,j}^B.\end{displaymath}
Hence, there is at least one interval $I_p$ such that
\begin{displaymath}\sum_{j_i\in S}t_{i,p}^A<\sum_{j_i\in S}t_{i,p}^B.\end{displaymath}
This gives that $t_{k,p}^A<t_{k,p}^B$ for some job $j_k\in S$. Thus, $t_{k,p}^A<|I_p|$ and $t_{k,p}^B>0$. If we consider any interval $I_j$, the sum of processing times of all jobs in $I_j$ is the same for all schedules satisfying lemma 1. So, there must be a job $j_{\ell}\notin S$ such that $t_{\ell,p}^A>t_{\ell,p}^B$. Therefore, $t_{\ell,p}^A>0$ and $t_{\ell,p}^B<|I_p|$. We conclude that $s_{\ell}^A\geq s_k^A>s_k^B\geq s_{\ell}^B$, which contradicts the fact that $j_{\ell}\notin S$.
\end{proof}
\newline

Notice that Lemma \ref{KKT-lem} does not explain how to find
an optimal schedule. The basic reason is that it does not determine
the speed value of each job. Moreover,
it does not specify exactly the structure of the optimal schedule.
That is, it does not specify which job is executed by each processor at each time $t$.

\section{An Optimal Combinatorial Algorithm}
In this section, we propose a combinatorial algorithm for our problem which always constructs a schedule satisfying the properties stated in the previous section. Our algorithm is based on the notion of {\em critical jobs} defined below.
The basic idea is to continuously decrease the speeds of jobs step by step. At each step, we assign a speed to the critical jobs that we
ignore in the subsequent steps and
we continue with the remaining subset of jobs.
At the end of the last step, every job has been assigned a speed.
In order to recognize the critical jobs, we consider a reduction
to the {\em Work Assignment Problem} (WAP).

Let us first give some notations and definitions concerning the maximum flow and minimum cut problems.
Consider a graph $G=(V,E)$  in which each edge $(u,v)$ has capacity $c(u,v)$ and two nodes $s,t\in V$. An $(s,t)$-cut of $G$ is a partition of its nodes into two disjoint subsets $X$ and $Y$ so that if we remove the edges $(u,w)$ with $u\in X$ and $w\in Y$, the nodes $s$ and $t$ are disconnected, i.e. there is no path from $s$ to $t$. A minimum $(s,t)$-cut $(X,Y)$ is a cut whose sum of the capacities of the edges $(u,w)$ with $u\in X$ and $w\in Y$ is minimized. In the following, we will consider an $(s,t)$-cut as the set of these edges. Also, given an $(s,t)$-flow of a graph $G=(V,E)$, we use the term $f(e)$ to denote the amount of flow that passes through the edge $e\in E$.

Given a graph $G$ and a flow $\mathcal{F}$, we define the residual graph $G_f$ of $G$ with respect to $\mathcal{F}$ as follows: (i) $G_f$ has the same set of nodes with $G$, (ii) for each edge $(u,v)$ in $G$ on which $f(u,v)<c(u,v)$, we include the edge $(u,v)$ with capacity $c(u,v)-f(u,v)$, and (iii) for each edge $(u,v)$ with $f(u,v)>0$, we include the edge $(v,u)$ with capacity $f(u,v)$. Next, we define the notion of upstream nodes that we will need throughout our analysis. A node $v$ is {\em upstream} if, for all min $(s,t)$-cuts $(X,Y)$, $v$ belongs in $X$. That is, $v$ lies on the source side of every min cut.

Now, for each instance of the WAP, we define a graph so as to reduce our original problem to the maximum flow problem.
Given an instance $<\mathcal{J},\mathcal{I},v>$ of the WAP, consider the graph $G=(V,E)$ that contains one node $x_i$ for each job $j_i$, one node $y_j$ for each interval $I_j$, a source node $s$ and a destination node $t$. We introduce an edge $(s,x_i)$ for each $j_i\in \mathcal{J}$ with capacity $\frac{w_i}{v}$, an edge $(x_i,y_j)$ with capacity $|I_j|$ for each
couple of $j_i$ and $I_j$ such that $j_i\in A(j)$ and an edge $(y_j,t)$ with capacity $m_j|I_j|$ for each interval $I_j\in\mathcal{I}$. We say that this is the corresponding graph of $<\mathcal{J},\mathcal{I},v>$.

At this point, we are ready to introduce the notion of criticality.
Given a feasible instance for the WAP, we say that job $j_i$ is {\em critical} if and only if for any feasible schedule and for each $I_j\subseteq span_i$, either $t_{i,j}=|I_j|$ or $\sum_{j_i\in A(j)}t_{i,j}=m_j|I_j|$. Furthermore, we say that an instance $<\mathcal{J},\mathcal{I},v>$ of the WAP is critical if and only if $v$ is the minimum speed so that the set of jobs $\mathcal{J}$ can be feasibly executed over the intervals in $\mathcal{I}$. With respect to graph $G$, a job $j_i$ is critical if and only if for any maximum flow, either the edge $(x_i,y_j)$ or the edge $(y_j,t)$ is saturated for each  $I_j$ such that $j_i\in A(j)$.
Notice that job  $j_i$ is also
critical for the $<\mathcal{J},\mathcal{I},v-\epsilon>$, for any $\epsilon >0$.
 
\newpage

\subsection{Properties of the Work Assignment Problem}

Next, we will prove some lemmas that will guide us to an optimal algorithm.
Our algorithm will be based on a reduction of our problem to
the maximum flow problem which is a consequence of the following
lemma.

\begin{lemma} \cite{BNS04} There exists a feasible schedule for the work assignment problem if and only if the corresponding graph has maximum $(s,t)$-flow equal to $\sum_{i=1}^n\frac{w_i}{v}$.\end{lemma}

At this point, we state a Lemma concerning the upstream nodes that we will need in one
of the proofs that follow. Also, for completeness, we present a proof that can be also be
found in \cite{KT06}.

\begin{claim} \cite{KT06}
The set of upstream nodes is reachable from the source node $s$ in the residual
graph of any maximum flow and therefore they can be found by
performing a breadth-first-search (BFS) starting from $s$.
\end{claim}

\begin{proof}
Let $(X,Y)$ the cut found after performing
a BFS on the residual graph $G_f$, starting from
the source $s$, at the end of any maximum
flow algorithm. If a  node $v$ is upstream then it must belong
to $X$. Conversely, assume that $v\in X$ and $v$ is not an
upstream node. This means that there is a cut $(X',Y')$ with
$v\in Y'$. Given that $v\in X$, there is a path $P$
from $s$ to $v$. Since $v \in Y'$, $P$ must have an edge
$(u,w)$ with $u$ in $X'$ and $w \in Y'$. However this is a
contradiction since there is an edge in $G_f$ that goes
from the source side to the sink side of a minimum cut.
\end{proof}

The following lemmas
that involve the notions of {\em critical job}
and {\em critical instance} are important ingredients for the analysis
of our algorithm.

\begin{lemma} \label{OneCrit}  If $<\mathcal{J},\mathcal{I},v>$ is a
critical instance of WAP, then there is at least one critical job $j_i\in\mathcal{J}$. \end{lemma}

\begin{proof}
Let $G$ be the graph that corresponds to a critical instance $<\mathcal{J},\mathcal{I},v>$, and let $G'$ be the graph that corresponds to the instance $<\mathcal{J},\mathcal{I},v-\epsilon>$, for a small constant $\epsilon>0$ that approaches zero. Since $<\mathcal{J},\mathcal{I},v>$ is critical, there is no feasible $(s,t)$-flow equal to $\sum_{j_i\in\mathcal{J}}\frac{w_i}{v-\epsilon}$ in $G'$. Because of the $max$ $flow-min$ $cut$ theorem, we can conclude that any minimum $(s,t)$-cut of $G'$ has capacity strictly less than $\sum_{j_i\in\mathcal{J}}\frac{w_i}{v-\epsilon}$ and as a result, there is no minimum $(s,t)$-cut of $G'$ that includes all edges $(s,x_i)$. If all edges $(s,x_i)$ were included in a minimum $(s,t)$-cut, then $G'$ would have an $(s,t)$-flow in which all these edges would be saturated which implies that there would be a feasible $(s,t)$-flow for $G'$ with value $\sum_{j_i\in\mathcal{J}}\frac{w_i}{v-\epsilon}$.


The remainder of the proof is based on the notion of {\em upstream nodes}.
For that, it suffices to observe that given any maximum flow,
there is always an edge $(s,x_i)$
that is not saturated.
Firstly, we need to show that there is
always an $x_i$ node in $G'$ which belongs to the set
of upstream nodes.
If we apply breadth-first search on the residual graph $G_f$,
we will reach $x_i$ which implies that $x_i$
is upstream. Thus, for every path $x_i,y_j,t$ of $G'$, there is always
  an edge
$(x_i,y_j)$ or $(y_j,t)$ that is saturated
by any maximum flow. This holds since if not, there would be
an unsaturated $(s,t)$ path (a path is {\em saturated} if at least one of its edges is saturated) contradicting the maximality of the flow.
Hence, $j_i$, the job that corresponds to $x_i$,
 is a critical job.
\end{proof}

\begin{lemma} \label{AllCrit}
Let $G=(V,E)$ be the graph that corresponds to the instance $<\mathcal{J},\mathcal{I},v>$ of the WAP. If the edge $(y_j,t)\in E$ belongs to a minimum $(s,t)$-cut of $G$ and there is a maximum $(s,t)$-flow such that $f(x_i,y_j)>0$, then $j_i$ is critical.
\end{lemma}
\begin{proof} Suppose that the edge $(y_j,t)$ belongs to a minimum $(s,t)$-cut $\mathcal{C}$ and that there is a maximum $(s,t)$-flow $\mathcal{F}$ such that $f(x_i,y_j)>0$. $\mathcal{C}$ is saturated by any maximum flow. Since $f(x_i,y_j)>0$, it is not possible that a path from $x_i$ to $t$ is left unsaturated  by $\mathcal{F}$ because if this was the case, then
we could send part of $f(x_i,y_j)$ through the unsaturated path and this would contradict the fact that $(y_j,t)$ belongs to a minimum $(s,t)$-cut. Since $\mathcal{F}$ is a maximum $(s,t)$-flow and saturates all the paths from $x_i$ to $t$, there should be a minimum $(s,t)$-cut $\mathcal{C}'$ that contains one edge from each such path (the one that is saturated by $\mathcal{F}$). Hence, $j_i$ is critical.
\end{proof}
\newline

Our algorithm is based on the following lemma in order to determine critical jobs.
\begin{lemma}
Assume that $<\mathcal{J},\mathcal{I},v>$ is a critical instance for WAP
and let $G'$ be the graph that corresponds to the instance
$<\mathcal{J},\mathcal{I},v-\epsilon>$. Then, any minimum $(s,t)$-cut $\mathcal{C}'$ of $G'$ contains:
\begin{itemize}
\item[(i)] exactly one edge of every path $x_i, y_j,t$ for any critical job
$j_i$ of $G$,
\item[(ii)] all the $(s,x_i)$ edges for any non-critical job $j_i$
of $G$.
\end{itemize}
\end{lemma}

\begin{proof}
Consider any critical job $j_i$. Assume that there is a path
$x_i, y_j,t$ in $G'$ such that
none of its edges belong to a minimum $(s,t)$-cut
$\mathcal{C}$. Then there
is a maximum $(s,t)$-flow $\mathcal{F}$ that does not saturate the edges
$(x_i,y_j)$ and $(y_j,t)$.
If the edge $(s,x_i)$ was not saturated, then  $\mathcal{F}$
would not be a maximum flow. On the other hand, if  $(s,x_i)$ was
saturated by $\mathcal{F}$, then job $j_i$ would not be critical
for  $<\mathcal{J},\mathcal{I},v>$. In both cases, we have a contradiction.

Similarly, assume that $j_i$ is not critical for the instance
$<\mathcal{J},\mathcal{I},v>$ and suppose that the edge $(s,x_i)$ does not
belong to a minimum cut of G'. This means that there is a maximum $(s,t)$-flow
$\mathcal{F}$ that does not saturate this edge. If there is at least one path
$x_i,y_j,t$ that is not saturated, then $\mathcal{F}$ is not maximum
and if all paths are saturated then $j_i$ is a critical job for
 $<\mathcal{J},\mathcal{I},v>$, which is a contradiction.
\end{proof}

\newpage

\subsection{The BAL Algorithm}

We are now ready to give a high level description of our algorithm. Initially, we will assume that the optimal schedule consumes an unbounded amount of energy and we assume that all jobs are executed with the same speed $s_{UB}$. This speed is such that there exists a feasible schedule that executes all jobs with the same speed. Then, we decrease the speed of all jobs up to a point where no further reduction is possible so as to obtain a feasible schedule. At this point, all jobs are assumed to be executed with the same speed, which is critical, and there is at least one job that cannot be executed with speed less than this. The jobs that cannot be executed with speed less than the critical one form the current set
of critical jobs. So, the critical job(s) is (are)
assigned the critical speed and is (are) ignored after this point. That is, in what follows, the algorithm considers the subproblem in which some jobs are omitted (critical jobs), because they are already assigned the lowest speed possible (critical speed) so that they can be feasibly executed, and there are less than $m$ processors during some intervals because these processors are dedicated to the omitted jobs (i.e. we get an instance of WAP).
Our algorithm can be described as follows:
\begin{algorithm}[h!] \nonumber
\caption{BAL}
\label{alg1}
\begin{algorithmic}[1]
\STATE $s_{UB}=\max\{\max_j\{\frac{\sum_{j_i \in A(j)} w_i}{|I_j|}\},\max_{j_i}\{den_i\}\}$, $s_{LB}=\max_{j_i\in\mathcal{J}}\{den_i\}$
\WHILE {$\mathcal{J}\neq\emptyset$}
\STATE Find the minimum speed $s_{crit}$ so that the instance $<\mathcal{J}, \mathcal{I}, s_{crit}>$ of the WAP problem  is feasible, using binary search in the interval $[s_{LB},s_{UB}]$, through
repeated maximum flow computations.
\STATE Determine the set of critical jobs $\mathcal{J}_{crit}$.
\STATE Assign to the critical jobs speed $s_{crit}$ and set $\mathcal{J}=\mathcal{J}\backslash\mathcal{J}_{crit}$.
\STATE Update $\mathcal{I}$, i.e., the number of available processors $m_j$ for each interval
$I_j$.
\STATE $s_{UB}=s_{crit}$, $s_{LB}=\max_{j_i\in\mathcal{J}}\{den_i\}$
\ENDWHILE
\STATE Use the optimal algorithm for $P|r_i,d_i,pmtn|-$ to schedule each job with processing time $w_i/s_i$.
\end{algorithmic}
\end{algorithm}

We denote $s_{crit}$ the critical speed and $\mathcal{J}_{crit}$ the set of critical jobs. We know that each job will be executed with speed not less than its density. Therefore, given a set of jobs $\mathcal{J}$, we know that there does not exist a feasible schedule that executes all jobs with the speed $s<\max_{j_i\in\mathcal{J}}\{den_i\}$. Also, observe that
no job has speed $s>\max\{\max_j\{\frac{\sum_{j_i \in A(j)} w_i}{|I_j|}\},\max_{j_i}\{den_i\}\}$. These bounds define the search space of the binary search for the first step of the algorithm in order to determine the minimum speed for which there is a feasible schedule that executes all jobs in $\mathcal{J}$ with the same speed. In the subsequent
step the current speed (i.e. the critical speed of the previous step) is an upper bound on the speed of all remaining
jobs and a lower bound is the maximum density among them. We
use these updated bounds  to perform
a new binary search and we go on like that. At this point, note that binary search has already been
used in other works as part of  optimal polynomial-time
algorithms for scheduling problems with speed scaling (see \cite{AF07} and \cite{PUW08}).

In order to complete the description of our algorithm, it remains to
explain the way critical jobs are determined. Because of Lemma 5, this can be done by
finding a minimum $(s,t)$-cut in the graph $G'$ that corresponds to $<\tilde{\mathcal{J}},\tilde{\mathcal{I}},v-\epsilon>$
where $\tilde{\mathcal{J}}$ and $\tilde{\mathcal{I}}$ correspond to the current instance of the WAP.
Note that $\epsilon$ must be such that $v-\epsilon$ is strictly greater than the next critical speed.

Algorithm BAL produces an optimal schedule, and this holds because any schedule constructed by the algorithm satisfies the properties of Lemma \ref{KKT-lem}.

\begin{theorem} Algorithm BAL  produces an optimal schedule. \end{theorem}

\begin{proof}
First of all, it is obvious that the algorithm assigns to every  job a constant speed because each job is assigned exactly one speed in one iteration.
Because of Lemma \ref{AllCrit}, we know that all jobs that have $0<t_{i,j}<|I_j|$ will have the same speed because when such a job is critical all other jobs
of the same kind are critical as well and are assigned the same speed. For the same reason, each job with $t_{i,j}=|I_j|$ will be assigned the same speed with all jobs that will run during $I_j$ or a greater one in a previous step.

Now, consider the case where $s_i=0$ for a job $j_i$ during an interval $I_j\subseteq span_i$. When $j_i$ is assigned a speed by the algorithm, it is critical. Hence, in every interval $I_j$ such that $j_i$ is alive, apart from the ones
whose processors were already occupied in previous iterations, we know that either $t_{i,j}=|I_j|$ or $\sum_{j_i\in A(j)}t_{i,j}=m_j|I_j|$, where $m_j$ is the number of the available processors. Therefore, if $t_{i,j}=0$ then there
are two cases: either (i)  $I_j$ had all its processors occupied in a previous iteration than the one that $j_i$ was assigned a speed, or (ii) this happened at the same iteration and the minimum speed that a job has during this interval is not less than the one of $j_i$. Hence, $j_i$ cannot get greater speed than any job executed during $I_j$.
Finally, because of Lemma 5, BAL correctly identifies the critical jobs
at each step of the algorithm. The theorem follows.
\end{proof}

We turn, now, our attention to the complexity of the algorithm. Because of Lemma \ref{OneCrit} at least one job (all critical ones) is scheduled at each step of the algorithm. Therefore, there will be at most $n$ steps. Assume that $P$ is the range of all possible
values of speeds divided by our desired accuracy. Then binary search, needs
$O(\log P)$ values of speed to determine the next critical speed at one step. That is,
BAL performs $O(\log P)$ maximum flow calculations at each step.
Thus,  the overall complexity
of our algorithm is   $O(nf(n)\log P)$.
\newline
\newline
\noindent
{\bf Relation of BAL with  the algorithm of Albers et al. \cite{AAG11}.}
The high-level idea of the algorithm in \cite{AAG11} is similar
 with the one of BAL.
Both algorithms can be decomposed  in
a number of steps (phases) and at each step,
a subset of jobs (the critical ones) is scheduled.
The difference between the two
algorithms is the way each step is performed.
In \cite{AAG11}, a step is as follows:
at the beginning, all remaining jobs
are conjectured to be critical.
Then, the set of (potential) critical jobs is reduced
through repeated maximum flow computations.
Once the set of critical jobs of a particular step
is determined, their algorithm specifies
the way these jobs are executed. In the worst case, their algorithm
performs $n$ steps and the $i$-th step involves $n-i$ maximum
flow computations.
Therefore, the worst-case running time of their algorithm
is $O(n^2f(n))$.
In our case, BAL computes the speed of critical jobs through
binary search. Each iteration of the binary search
involves a  maximum flow computation.
Once the critical speed is computed, the set of critical
jobs can be found by computing a minimum-cut.
BAL constructs the schedule once all the critical speeds are determined.

\section{Makespan Minimization with a Budget of Energy}

Algorithm BAL can be extended to obtain an optimal algorithm, say MBAL, for the problem of makespan minimization given a fixed budget of energy $E$. As before, preemption and migration are allowed and jobs have arbitrary release dates and
works. In order to apply MBAL, we will need an upper and a lower bound on the makespan of the optimal schedule. Then, the algorithm uses binary search to compute the minimum makespan for which there is a feasible schedule consuming $E$ units of energy. Two such bounds are $X_{LB}=\frac{1}{m}(\frac{W^{\alpha}}{E})^{\frac{1}{\alpha-1}}$ and
$X_{UB}=\max_i\{r_i\}+ (\frac{W^{\alpha}}{E})^{\frac{1}{\alpha-1}}$  where $W$ is the total work
of all jobs. The high-level description of the algorithm is the following:
\begin{algorithm}[h!]
\caption{MBAL}
\label{alg2}
\begin{algorithmic}[1]
\STATE Compute $X_{UB}$ and $X_{LB}$.
\STATE Perform binary search in $[X_{LB},X_{UB}]$ to find the minimum makespan $X^*$ for which there is a feasible schedule that consumes an $E$ amount of energy.
\STATE Return this schedule.
\end{algorithmic}
\end{algorithm}

In order to perform the binary search, given a value $X$, MBAL examines whether or not there is a feasible schedule of makespan $X$ that consumes $E$ units of energy. To do this, it runs algorithm BAL assuming that all jobs have a common deadline $X$. Then, it computes the minimum value of energy $E^*$ that a feasible schedule for the particular instance might have. If $E\geq E^*$, then there is a feasible schedule that executes the jobs using no more than $E$ energy with makespan $X$. Otherwise, there does not exist such a schedule.
The complexity of MBAL is $\log P$ times the complexity of BAL, i.e. $O(nf(n)\log^2 P)$.

\section{Conclusion}

We studied the energy minimization
multiprocessor speed scaling problem with migration.
We proposed a combinatorial polynomial time algorithm
based on a reduction  to the maximum flow problem.
We also extended our result in the case where the
objective is makespan minimization given a budget of
energy.
Since there is not much work on problems with migration there are many directions and problems to be considered for multicriteria optimization. All these problems seem to be very interesting and might require new algorithmic techniques because of their continuous nature. In this context, we believe that the approach used in our paper may be useful for future works.
\newline
\newline
\textbf{Acknowledgments}
\newline
We thank Alexander Kononov for helpful discussions on this work.

\end{document}